\documentclass[conference]{IEEEtran}
\IEEEoverridecommandlockouts
\usepackage{cite}
\usepackage{amsmath,amssymb,amsfonts}
\usepackage{amsthm}
\usepackage{algorithmic}
\usepackage{graphicx}
\usepackage{textcomp}
\usepackage{xcolor}

\usepackage{url}
\usepackage{caption}
\usepackage{subcaption}
\usepackage{xspace}
\usepackage[normalem]{ulem}
\usepackage{multirow}
\usepackage{color}

\usepackage{etoolbox}

\usepackage{bm}

\def\eqref#1{equation~\ref{#1}}

\def\1{\bm{1}}

\def\vzero{{\bm{0}}}

\def\vf{{\bm{f}}}

\def\vu{{\bm{u}}}
\def\vv{{\bm{v}}}
\def\vw{{\bm{w}}}

\def\vy{{\bm{y}}}

\def\evf{{f}}

\def\evu{{u}}
\def\evv{{v}}
\def\evw{{w}}

\def\mA{{\bm{A}}}

\def\mC{{\bm{C}}}

\def\mI{{\bm{I}}}

\def\mM{{\bm{M}}}

\def\mR{{\bm{R}}}
\def\mS{{\bm{S}}}

\DeclareMathAlphabet{\mathsfit}{\encodingdefault}{\sfdefault}{m}{sl}
\SetMathAlphabet{\mathsfit}{bold}{\encodingdefault}{\sfdefault}{bx}{n}

\def\emC{{C}}

\def\emR{{R}}
\def\emS{{S}}

\newcommand{\eat}[1]{}
\newcommand{\mnote}[2][FIXME]{\textcolor{red}{[\textbf{#1:}} {\color{blue} {#2}}\textcolor{red}{\textbf{]}}}

\newcommand{\stitle}[1]{\vspace{0.1cm}\noindent\textbf{#1}}

\newcommand{\method}{Att\-Rank\xspace}
\newcommand{\ourvec}{attention\xspace}

\newtheorem{theorem}{Theorem}
\newtheorem{problem}{Problem}

\IEEEoverridecommandlockouts

\begin{document}

\title{Ranking Papers by their Short-Term Scientific Impact$^*$\thanks{$^*$A short version of this paper appears in \cite{kanellos2021ranking}.}}

\author{\IEEEauthorblockN{Ilias Kanellos}
\IEEEauthorblockA{\textit{Athena R.C.}\\
Athens, Greece}
\IEEEauthorblockA{ilias.kanellos@athenarc.gr}\\
\and
\IEEEauthorblockN{Thanasis Vergoulis}
\IEEEauthorblockA{\textit{Athena R.C.}\\
Athens, Greece}
\IEEEauthorblockA{vergoulis@athenarc.gr}
\and
\IEEEauthorblockN{Dimitris Sacharidis}
\IEEEauthorblockA{\textit{ULB}\\
Brussels, Belgium}
\IEEEauthorblockA{dimitris.sacharidis@ulb.be}
\and
\IEEEauthorblockN{Theodore Dalamagas}
\IEEEauthorblockA{\textit{Athena R.C.}\\
Athens, Greece}
\IEEEauthorblockA{dalamag@athenarc.gr}
\and
\IEEEauthorblockN{Yannis Vassiliou}
\IEEEauthorblockA{\textit{NTUA}\\
Athens, Greece}
\IEEEauthorblockA{yv@cs.ntua.gr}
}

\maketitle

\begin{abstract}
The constantly increasing rate at which scientific papers are published makes it difficult for researchers to identify papers that currently impact the research field of their interest. Hence, approaches to effectively identify papers of high impact have attracted great
attention in the past. In this work, we present a method that seeks to rank papers based on their estimated short-term impact, as measured by the number of citations received in the near future. Similar to previous work, our method models a researcher as she explores the paper citation network. The key aspect is that we incorporate
an attention-based mechanism, akin to a time-restricted version of preferential attachment,
to explicitly capture a researcher's preference to read papers which received a lot of attention recently. A detailed experimental evaluation on four real citation datasets across disciplines, shows
that our approach is more effective than previous work in ranking papers based on their short-term impact.
\end{abstract}


\section{Introduction}

Quantifying the importance of scientific publications, colloquially called \emph{papers}, is an important research problem with various applications. For example, a student that wants to familiarize herself with a research area, may look for seminal papers in the field. A hiring committee may assess an applicant based on the 
aggregate impact of his publication record. 
As the number of papers published 
grows at an increasing 
rate~\cite{larsen2010rate,bornmann2015growth}, 
discerning the important papers, especially among the recent publications, becomes a hard task.

Conventionally, the importance of a paper 
is bestowed upon itself by its peers, the
subsequent papers that continue the line 
of research and acknowledge its 
contribution by citing it. Therefore, 
the scientific impact of a paper depends on the network of citations.
In this work, we focus on the 
\emph{short-term impact} (STI) of a 
paper, quantified by 
the number of citations it 
acquires
in the near future (
referred to as ``new citations''
\cite{walker2007ranking}, or 
``future citation counts''
\cite{ghosh2011time}). 
Specifically, 
we address the research problem of
ranking papers via their expected 
short-term impact. 

Existing work typically assigns to each paper a proxy score estimating its expected short-term impact. These scores are determined by a stochastic process, akin to PageRank \cite{page1999pagerank}, modelling the 
impact flow in the citation network. The important concern here is to account for the \emph{age bias} inherent in citation networks
~\cite{chen2007finding,hwang2010yet,yu2005adding,liao2017ranking}: as papers can only cite past work, recent publications are at a disadvantage having 
less opportunity to accumulate citations.
A popular way to address this is by
 introducing time-awareness into the stochastic process, by favoring either recent papers or recent citations~\cite{walker2007ranking,sayyadi2009futurerank,ghosh2011time}. Nonetheless, it has been shown~\cite{KVSDV2019} that these methods still leave enough space for further improvements.

In this paper, we argue that there is an additional, previously
unexplored mechanism that governs
where future citations end up. We posit that recent citations strongly influence the short-term impact, in that the level of \emph{attention} papers currently enjoy will not
change significantly in the very 
near future. We investigate this hypothesis and find that 
it holds to a certain degree across different citation networks. 
Hence, we introduce an attention-based mechanism, reminiscent of a 
time-restricted version of preferential
attachment~\cite{barabasi2014network}, 
that models the fact that recently 
cited papers 
continue getting cited in the short-term.

The proposed paper ranking method, called \method, describes an iterative process simulating a researcher 
reading existing literature.
At each step in the process, the researcher has studied some 
paper and decides what to read next among three options: 
(a) pick a reference from the current paper, 
(b) pick a recent paper, and 
(c) pick a currently popular paper. 
The first option models the impact flow of
impact from citations, 
the second option mitigates age bias, while the third option 
models the aforementioned attention-based mechanism of network
growth. 
We can guarantee that, if the probabilities are properly configured, this process will always converge (see Theorem~\ref{th:convergence}).
This converged \method score of each paper acts as a proxy 
to its unknown short-term impact. Hence, 
to estimate their STI ranking
we rank papers in decreasing order of their \method score.

To evaluate \method's effectiveness in identifying papers 
with high short-term impact, we perform an extensive experimental
evaluation, 
on four citation networks 
from various scientific disciplines. 
We measure effectiveness as the ranking accuracy with respect to 
the ground truth STI ranking. We investigate the 
importance of the attention mechanism in 
achieving high effectiveness.
We also compare \method against several 
state-of-the-art methods, which are carefully tuned for each experimental setting. 
Our findings indicate 
 that across almost all settings, 
\method 
outperforms prior work.

The contributions of this work are as follows:
\begin{itemize}
    \item We study the problem of ranking papers by their short-term impact (STI), and observe that among the top ranking papers we not only find papers published recently, but also papers that have just recently become popular.
    
    \item We propose a popularity-based model of growth for the citation network that seeks to explain the aforementioned observation. We then introduce paper ranking method, called \method, that materializes this model.

    \item We perform an extensive experimental evaluation that highlights the importance of the popularity-based growth mechanism in achieving superior performance against state-of-the-art methods. Specifically, we find that \method achieves higher positive correlations to 
    STI rankings, by up to $0.077$ units\eat{\color{cyan}
    ($ 4.81\% $)} 
    compared to its competitors, and higher nDCG values, by up to 
    $ 0.098 $ units.
    
    \item \method's implementation is 
    scalable and can be executed on 
    very large citation networks. It 
    will be made publicly available 
    under the GNU/GPL license.
    
\end{itemize}

\noindent \textbf{Outline.} The remaining of this paper is structured as follows. In Section~\ref{sec-back}, we introduce the problem and related concepts.
In Section~\ref{sec-method}, we investigate the attention-based mechanism and introduce our method, \method. Then, in
Section~\ref{sec-exp}, we 
experimentally evaluate \method's 
effectiveness in comparison to the
state-of-the-art paper ranking methods.
In Section~\ref{sec-related},
we discuss related work. We conclude our
contribution in Section~\ref{sec-concl}.

\section{Background}
\label{sec-back}

\stitle{Citation Network.}
We represent a collection of
papers as a directed graph,
which we call the
\emph{citation network}. Each node in
this graph corresponds to a paper, while each
directed edge corresponds to a citation.

A citation network can be represented
by its \emph{citation matrix}
$ \mC $, where entry $ \emC_{i,j} = 1 $,
iff paper $j $ cites paper $ i $,
and $ \emC_{i,j} = 0 $, otherwise.
We follow the convention that paper ids indicate the order in which they are published; i.e., paper $i$ was published before $j$ iff $i<j$.

To distinguish between different states of the citation network as it evolves over time, we use the \emph{superscript parenthesis} notation $(t)$ to refer to a state where only papers with id up to $t$ have been published. For example, $\mC^{(t)}$ is the $t \times t$ citation matrix for the first $t$ papers. Because of how the citation matrix evolves, $\mC^{(t)}$ is a submatrix of $\mC^{(t')}$ when $t<t'$.
We use the \emph{subscript bracket} notation $[n]$ to refer to a submatrix containing the first $n$ rows and columns, or to a subvector containing the first $n$ entries; $[-n]$ denotes the last $n$ rows/columns/entries.
For example, the previous observation can be expressed as $\mC_{[n]}^{(t)} = \mC_{[n]}^{(t')}$ for any $n \le t < t'$.

\stitle{PageRank.}
PageRank \cite{page1999pagerank} measures the importance of a node in a network, by defining a random walk with jumps process.
In the context of citation networks, the process simulates 
a ``random researcher'', who starts their work by reading a paper. Then, with probability $ \alpha$, they pick another paper to read from the reference list, or, with
probability $1-\alpha$, choose any other paper in the network at random.

Given a citation matrix $\mC$, define the stochastic matrix $\mS$ as follows.
Let $ k_i $ denote the number of papers referenced by $ i $. Then,
$\emS_{i,j} = \frac{1}{k_j} $,
iff paper
$ j $ cites paper $ i $,
$ \emS_{i,j} = 0 $, iff $ j $ does not
cite $i$ but cites at least
one other paper,
and $ \emS_{i,j} = \frac{1}{n} $,
iff paper $ p_j $ cites no paper (i.e., is a dangling node), where $n$ is the number of papers.

Let $\vu$ denote the \emph{teleport vector},\footnote{All vectors are column vectors.} such that $|\vu| \equiv \sum_i \evu_i = 1$ and $\forall i$ $\evu_i \ge 0$; typically, $\vu$ is defined to indicate uniform teleport probabilities, i.e., $\forall i$ $\evu_i = 1/n$.
Let $\alpha \in [0,1)$ denote the \emph{random jump probability}.
Then the \emph{PageRank vector} $\vv$ is defined by this equation:
\begin{equation}
\vv = \alpha \mS \vv + (1-\alpha) \vu.
\label{eqn:pagerank}
\end{equation}

We say that $\vv$ is the PageRank vector of matrix $\mS$ with respect to teleport vector $\vu$.
The PageRank vector is given by the following closed-form formula:
\begin{equation}
\vv = (1-\alpha) \vu +  (1-\alpha) \sum_{x=1}^\infty \alpha^x \mS^x \vu =  (1-\alpha) \sum_{x=0}^\infty \alpha^x \mS^x \vu,
\label{eqn:pagerank-closed}
\end{equation}
where the convention $\mA^0 \equiv \mI$ is used in the last equality. If we define $\mM = (1-\alpha) \sum_{x=0}^\infty \alpha^x \mS^x$, then we observe the linear relationship between the PageRank and the teleport vectors: $\vv = \mM \vu$.

Computing the PageRank vector is not done by computing matrix $\mM$, but by iteratively estimating $\vv$ via Equation~\ref{eqn:pagerank}. Starting from some random values for $\vv$, satisfying $|\vv|=1$ and $\forall i$ $\evv_i \ge 0$, at each step we update the PageRank vector as $\vv \gets \alpha \mS \vv + (1-\alpha) \vu$. In other words, for a given $\vu$, we have a convenient method to estimate $\mM \vu$.

\stitle{Short-Term Impact.}
Using node centrality metrics, such as the PageRank or simply the number of citations, to capture the impact of a paper can introduce biases, e.g., against recent papers, and may render important papers harder to
distinguish~\cite{hwang2010yet,chen2007finding,yu2005adding}.
This is due to the inherent characteristics of citation networks: the 
references of a paper are fixed, and there is a delay between a paper's publication and its first citation, known as \emph{citation lag}~\cite{diodato1994dictionary}.
This phenomenon is best portrayed in Figure~\ref{fig:distro},
where it is shown that, for different citation networks
(introduced in Section~\ref{sec:exp:methodology}),
the bulk of citations comes a few years after the paper
is published.
In contrast, the \emph{short-term impact} \cite{KVSDV2019}, also called the number/count of new/future citations \cite{walker2007ranking,ghosh2011time}, of a paper looks into a future state of the citation network and reflects the level of attention (in terms of citations) a paper will receive in the near future.

\begin{figure*}[!t]
\label{fig:distro-and-example}
\centering
\begin{subfigure}[t]{0.3\textwidth}
  \includegraphics[width=\textwidth]{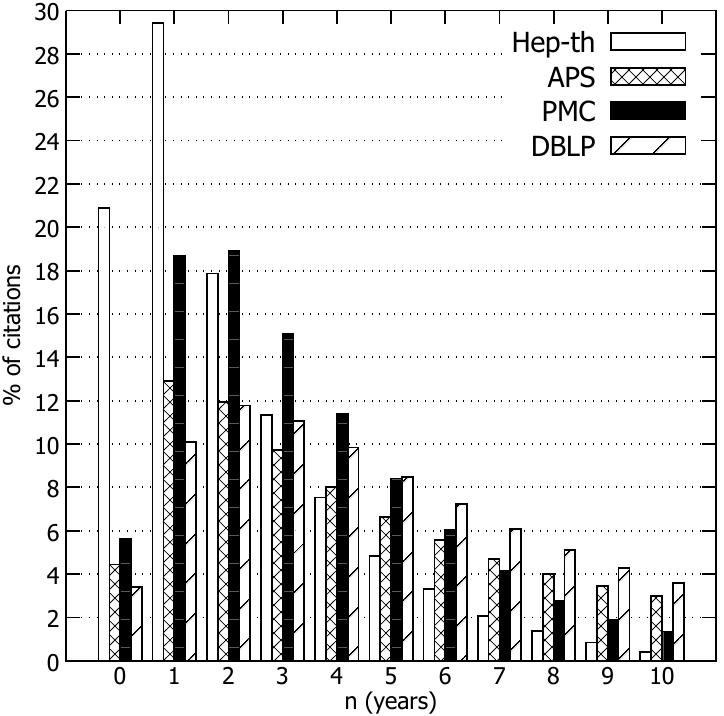}
  \caption{}
  \label{fig:distro}
\end{subfigure}%
\begin{subfigure}[t]{0.35\textwidth}
  \raisebox{0.6mm}{\includegraphics[width=\textwidth]{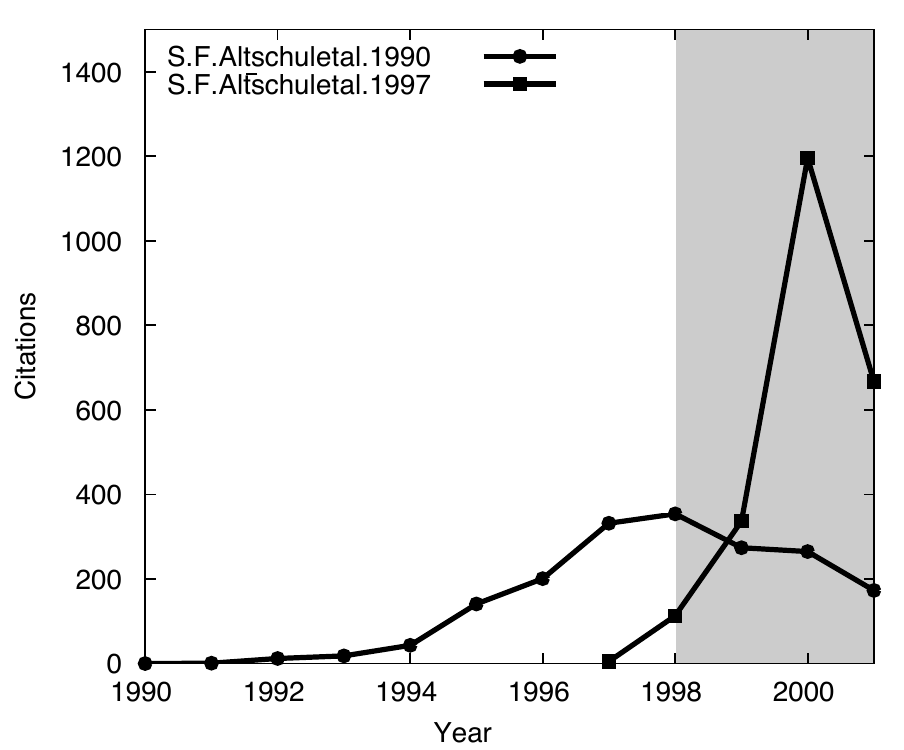}}
  \caption{}
  \label{fig:example}
\end{subfigure}
\caption{(a) Empirical
distribution of 
the fraction of total citations
received by papers $n$ years after
their publication
($n\leq10$), from four citation networks (see Section~\ref{sec:exp:methodology})
(b) A comparative yearly citation count of two papers.}
\end{figure*}

As a motivating example for the importance of
short-term impact, examine the case of two
seminal papers in the bioinformatics literature.
The first, published in 1990, introduces the initial version of the popular BLAST alignment algorithm, while the second, published in 1997, presents an improved alignment algorithm by the same team.
Figure~\ref{fig:example} comparatively illustrates
their yearly citation
counts.\footnote{Based on open citation
data from 
COCI 
(\url{http://opencitations.net/download}).}
Now, consider a bioinformatics researcher living in the year 1998. At that point in time, the older paper has a higher citation count. However, the newer paper is clearly more popular as it has a greater short-term impact, evidenced by the number of citations it collects in the next three years (highlighted in the figure).
The 1998 researcher would benefit from being
able to identify the newer paper as potentially having a higher
short-term impact.

Consider the state of the citation matrix at present time $n$.
Given a time \emph{horizon} of $\tau$, the \emph{short-term impact} (STI) $\evf_i$ of a paper $i$ (where $i \le n$) is defined as the number of future citations, i.e., those it would receive in the time period $(n,n+\tau]$:
\[
\evf_i = \sum_{j=n+1}^{n+\tau} \emC^{(n+\tau)}_{i,j}.
\]

Some observations are in order.
First, the time horizon $\tau$ is a user-defined parameter that specifies how long in the future one should wait for citations to accumulate. An appropriate value may depend on the typical duration of the research cycle (preparation, peer-reviewing, and publishing) specific to each scientific discipline.
Second, it is important to emphasize that STI can only be computed in retrospect; at current time $t_N$, the future citations are not yet observed. Thus, any method that seeks to identify papers with high STI has to employ a mechanism to account for the unobserved future citations.

With these remarks in mind and similar to prior work \cite{walker2007ranking,ghosh2011time,KVSDV2019}, we study the following problem.

\begin{problem}
Given the state $\mC^{(n)}$ of the citation network at current time $n$, return a ranking of papers such that it matches their ranking by short-term impact $\vf$ for a given time horizon $\tau$.
\end{problem}

\section{Approach}
\label{sec-method}

We first overview our approach presenting insight about how to rank based on STI. Then, we describe a useful tool, which we later exploit.

\subsection{Overview}

To rank by short term impact without knowing the future, one needs to introduce assumptions about how the citation network evolves.
In this work, we assume that the PageRank vector at time $t \ge n$ indicates the chance the papers will get cited at future time $t+1$.
This defines a mechanism, where a random researcher explores literature. While reading a paper, the researcher may choose with probability $\alpha$ to further read one of the paper's references, or pick another paper at random from some prior (teleport) probability.
So the probability of a paper $i$ receiving a citation from a random researcher at time $t+1$ is proportional to their PageRank $\evv^{(t)}_i$ at time $t$, which satisfies:
\[
\vv^{(t)} = \alpha \mS^{(t)} \vv^{(t)} + (1-\alpha) \vu^{(t)},
\]
where $\mS^{(t)}$ is the stochastic matrix and $\vu^{(t)}$ is the teleport vector at $t$.

The short-term impact of a paper is the sum of citations it will receive at each time $t$ between $n+1$ and $n+\tau$. Let us further assume that each paper in the future (after current time $n$) makes the same number of citations. So, under our assumptions, the number of future citations $\evf_i$ a paper $i$ receives is proportional to:
\[ \evf_i \propto  \sum_{t=0}^{\tau-1} \evv_{i}^{(n+t)}. \]

For the purposes of ranking, the scale of individual $\evf_i$ does not matter. Therefore, to rank by STI we would like to estimate the following vector:
\begin{equation}
\vy = \sum_{t=0}^{\tau-1} \vv_{[n]}^{(n+t)},
\label{eqn:target}
\end{equation}
i.e., the sum at each future timestamp of the PageRank of the present $n$ papers.

In Equation~\ref{eqn:target}, each PageRank vector $\vv_{[n]}^{(n+t)}$ is a random vector under the aforementioned citation network growth process, and its value depends on the values of the PageRank vector at previous times. Therefore, one way to estimate the mean of $\vy$ is with Markov Chain Monte Carlo methods to draw samples from the probability distribution of $\vy$. This however is costly, as each sample requires the computation of $\tau$ PageRank vectors for a large citation network with at least $n$ nodes.

We propose a different approach. We start by assuming that the future citation matrix and thus $\mS^{(n+\tau)}$ is known. Conditional to $\mS^{(n+\tau)}$, the PageRank vector $\vv_{[n]}^{(n+t)}$ become independent. We then apply a mechanism to rewrite Equation~\ref{eqn:target} so that it can be computed with a single PageRank computation. This rewriting now has quantities derived from the $\mS^{(n+\tau)}$ matrix. At the final step, we estimate these quantities from the current (at time $n$) state of the network, and compute the rewritten Equation~\ref{eqn:target}.

We next describe the mechanism that allows us to rewrite Equation~\ref{eqn:target} as a single PageRank computation.

\subsection{The Contraction Mechanism}
\label{sec:contraction}

Consider a network of $n+1$ nodes, where the $(n+1)$-th node has no incoming edges. Its stochastic matrix $\mS$ can be written as the following block matrix:
\[
\mS =
\begin{pmatrix}
\mS_{[n]} & \emS_{[n],n+1} \\
\vzero^\intercal & 0
\end{pmatrix},
\]
where $\mS_{[n]}$ is the submatrix of $\mS$ that contains the first $n$ rows and columns, vector $\emS_{[n],n+1}$ is the first $n$ rows of the $(n+1)$ column of $\mS$ and indicates the outgoing connections of node $n+1$, and $\vzero$ is the zero vector (here of size $n$).

Contraction provides the means to compute the PageRank vector for stochastic matrix $\mS$ by contracting node $n+1$ and computing a PageRank vector for $\mS_{[n]}$. To account for the contribution from node $n+1$ to its connections, the teleport probabilities to them are appropriately adjusted.

\begin{theorem}
\label{thm:contract}
Let $\vv = ( \vv_{[n]} \ \evv_{n+1})^\intercal$ denote the PageRank vector for stochastic matrix
$\mS = \begin{pmatrix}
\mS_{[n]} & \emS_{[n],n+1} \\
\vzero^\intercal & 0
\end{pmatrix}$
with respect to teleport vector $\vu = ( \vu_{[n]} \ \evu_{n+1})^\intercal$.
Define the adjusted teleport vector $\mathring\vu_{[n]} = \vu_{[n]} + \alpha \evu_{n+1} \emS_{[n],n+1}$.
Then, the normalized $\vv_{[n]}$ is the PageRank for matrix $\mS_{[n]}$ with respect to the normalized $\mathring\vu_{[n]}$, i.e., it holds that:
\[
\frac{1}{|\vv_{[n]}|} \vv_{[n]}  = \frac{\alpha}{|\vv_{[n]}|} \mS_{[n]} \vv_{[n]} + \frac{1-\alpha}{|\mathring\vu_{[n]}|}\mathring\vu_{[n]}.
\]
Moreover, $\evv_{n+1} = (1-\alpha) \evu_{n+1}$ and $|\vv_{[n]}| = |\mathring\vu_{[n]}| = 1 - \evv_{n+1}$.
\end{theorem}

\begin{proof}
Observe that the $x$-th power ($x \ge 1$) of $\mS$ is given by:
\[
\mS^x =
\begin{pmatrix}
\mS_{[n]}^x & \mS_{[n]}^{x-1}\emS_{[n],n+1} \\
\vzero^\intercal & 0
\end{pmatrix}.
\]
From Equation~\ref{eqn:pagerank-closed}, we have:

\begin{align*}
\begin{pmatrix} \vv_{[n]} \\ \evv_{n+1} \end{pmatrix} = & \
(1-\alpha) \begin{pmatrix} \vu_{[n]} \\ \evu_{n+1} \end{pmatrix} \\
+ & \ (1-\alpha) \sum_{x=1}^\infty \alpha^x
\begin{pmatrix}
\mS_{[n]}^x & \mS_{[n]}^{x-1}\emS_{[n],n+1} \\
\vzero^\intercal & 0
\end{pmatrix}
\begin{pmatrix} \vu_{[n]} \\ \evu_{n+1} \end{pmatrix}.
\end{align*}

From the bottom row, we immediately get $\evv_{n+1} = (1-\alpha) \evu_{n+1}$, while $|\vv_{[n]}| = 1 - \evv_{n+1}$.

We expand the top row to derive:
\begin{align*}
\vv_{[n]} =& \ (1-\alpha) \vu_{[n]} + (1-\alpha) \sum_{x=1}^\infty \alpha^x \mS_{[n]}^x \vu_{[n]} \\
& \  +  (1-\alpha) \sum_{x=1}^\infty \alpha^x \evu_{n+1} \mS_{[n]}^{x-1}\emS_{[n],n+1} \\
 =& \ (1-\alpha) \vu_{[n]} + (1-\alpha) \sum_{x=1}^\infty \alpha^x \mS_{[n]}^x \vu_{[n]} \\
& \  + (1-\alpha)\alpha \evu_{n+1} \emS_{[n],n+1} \\
& \ + (1-\alpha)\alpha \sum_{x=1}^\infty \alpha^x \evu_{n+1} \mS_{[n]}^x\emS_{[n],n+1} \\
 =& \ (1-\alpha) (\vu_{[n]} + \alpha \evu_{n+1} \emS_{[n],n+1} ) \\
 & \  +  (1-\alpha) \sum_{x=1}^\infty \alpha^x \mS_{[n]}^x (\vu_{[n]} + \alpha \evu_{n+1} \emS_{[n],n+1} )\\
 =& \ (1-\alpha) \mathring\vu_{[n]} +  (1-\alpha) \sum_{x=1}^\infty \alpha^x \mS_{[n]}^x \mathring\vu_{[n]}.
\end{align*}

By introducing $\mM_{[n]} = (1-\alpha) \sum_{x=0}^\infty \alpha^x \mS_{[n]}^x$, the previous equation is written as $\vv_{[n]} = \mM_{[n]} \mathring\vu_{[n]}$. Multiplying by $1/|\mathring\vu_{[n]}|$ we get that $\frac{1}{|\mathring\vu_{[n]}|} \vv_{[n]}$ is the PageRank vector for $\mS_{[n]}$ with respect to normalized $\mathring\vu_{[n]}$.
It remains to show that $|\mathring\vu_{[n]}| = |\vv_{[n]}|$, which holds because:
\begin{align*}
|\mathring\vu_{[n]}| =& \ |\vu_{[n]}| + \alpha \evu_{n+1} \sum_{i=1}^n \emS_{i,n+1} = 1 - \evu_{n+1} + \alpha \evu_{n+1}  \\
 =& \ 1 - (1-\alpha) \evu_{n+1} = 1 - \evv_{n+1} = |\vv_{[n]}|.
\end{align*}
\end{proof}

Contraction can be applied recursively to networks containing multiple nodes that have no incoming edges. Consider a network of $n+m$ nodes, represented by its stochastic matrix $\mS$, where its $m$ last nodes have no incoming edges. Given some teleport vector $\vu$ and random jump probability $\alpha$, let $\vv$ denote the PageRank vector that satisfies $\vv = (1-\alpha)\sum_{x=0}^\infty \alpha^x \mS^x \vu$. The PageRank scores $\vv_{[n]}$ of the first $n$ nodes can be computed directly from the PageRank on the stochastic matrix $\mS_{[n]}$ with respect to an adjusted teleport vector $\mathring\vu$, as indicated by the following theorem.

\begin{theorem}
\label{thm:contract_iter}
Let $\vv$ denote the PageRank vector for stochastic matrix $\mS =
\begin{pmatrix}
\mS_{[n]} & \mS_{[n],[-m]} \\
\vzero_{(m \times n)} & \vzero_{(m \times m)}
\end{pmatrix}$
with respect to teleport vector $\vu$ and random jump probability $\alpha$. Define the adjusted teleport vector $\mathring\vu_{[n]} = \vu_{[n]} + \alpha \sum_{i=1}^{m} \evu_{n+i} \emS_{[n], n+i}$. Then it holds that:
\[
\frac{1}{|\vv_{[n]}|} \vv_{[n]}  = \frac{\alpha}{|\vv_{[n]}|} \mS_{[n]} \vv_{[n]} + \frac{1-\alpha}{|\mathring\vu_{[n]}|}\mathring\vu_{[n]},
\]
where $|\mathring\vu_{[n]}|=|\vv_{[n]}| = 1 - (1-\alpha) \sum_{i=1}^{m} \evu_{n+i}$.
\end{theorem}

\begin{proof}
Follows from applying Theorem~\ref{thm:contract} $m$ times. At the first iteration, the adjusted teleport vector is:
\[
\vu_{[n+m-1]}^{(n+m-1)} = \vu_{[n+m-1]} + \alpha \evu_{n+m} \emS_{[n+m-1], n+m}.
\]
At the second iteration:
\begin{align*}
\vu_{[n+m-2]}^{(n+m-2)} =\ & \vu_{[n+m-2]}^{(n+m-1)} + \alpha \evu_{n+m-1} \emS_{[n+m-2], n+m-1}\\
=\ &  \vu_{[n+m-2]} + \alpha \evu_{n+m} \emS_{[n+m-2], n+m}\\
& + \alpha \evu_{n+m-1} \emS_{[n+m-2], n+m-1}\\
=\ &  \vu_{[n+m-2]} + \alpha \sum_{i=m-1}^m  \evu_{n+i} \emS_{[n+m-2], n+i}.
\end{align*}
At the $m$-th iteration, we get $\vu_{[n]}^{(n)} = \mathring\vu_{[n]}$.

Moreover, for each $i$, such that $1 \le i \le m$, we have $\evv_{n+i} = (1-\alpha) \evu_{n+i}$, which means $|\vv_{[n]}| = 1 - \sum_{i=1}^m \evv_{n+i} = 1 - (1-\alpha) \sum_{i=1}^{m} \evu_{n+i}$.

Finally, for each $i$, it holds that $\sum_{j=1}^{n+i} \emS_{j,n+i} =  \sum_{j=1}^{n} \emS_{j,n+i} = 1$, since node $n+i$ has outgoing edges only to nodes with id not greater than $n$. As a consequence, $|\mathring\vu_{[n]}|=|\vv_{[n]}|$.
\end{proof}

In other words, to compute the PageRank w.r.t.\ $\mS$ for the first $n$ nodes, we can use an adjusted teleport vector $\mathring\vu_{[n]}$ (after normalization) to compute the PageRank w.r.t.\ $\mS_{[n]}$, and then scale the result by $|\mathring\vu_{[n]}|$.

\subsection{The \method Method}

We will apply the contraction idea to compute each future PageRank vector $\vv_{[n]}^{(n+t)}$ as a PageRank vector of the current stochastic matrix $\mS^{(n)}$, assuming that its future state $\mS^{(n+t)}$ is known. Note that to apply the contraction idea, we need to restrict each future paper $i$ to cite no other future paper $j$ ($n<j<i$), i.e., citations only come for papers already published until current time $n$. This restriction only affects the PageRank values of the future papers, which however we do not need to rank.

Note that for any time $n+t$ where $t>0$, the first $n$ rows and columns of the stochastic matrix $\mS^{(n+t)}$ remain constant, and we denote $\mS \equiv \mS_{[n]}^{(n+t)}$. From Theorem~\ref{thm:contract_iter} we have:
\[
\frac{1}{|\vv_{[n]}^{(n+t)}|} \vv_{[n]}^{(n+t)}  = \frac{\alpha}{|\vv_{[n]}^{(n+t)}|} \mS \vv + \frac{1-\alpha}{|\vv_{[n]}^{(n+t)}|}\mathring\vu_{[n]}^{(n+t)}.
\]
Defining $\mM = (1-\alpha) \sum_{x=0}^\infty \alpha^x \mS^x$, we rewrite the previous equation in the closed form of Equation~\ref{eqn:pagerank-closed}:
\[
\vv_{[n]}^{(n+t)} = \mM \mathring\vu_{[n]}^{(n+t)}.
\]

From the definition of STI, we derive:
\[
\vy = \sum_{t=0}^{\tau-1} \vv_{[n]}^{(n+t)} =  \mM \sum_{t=0}^{\tau-1} \mathring\vu_{[n]}^{(n+t)},
\]
where $\mathring\vu_{[n]}^{(n+t)} = \vu_{[n]}^{(n+t)} + \alpha \sum_{i=1}^{m} \evu_{n+i}^{(n+t)} \emS_{[n], n+i}^{(n+t)}$.

For convenience, we further assume that the teleport vector for the first $n$ papers is the same at each time $n+t$, and we denote it as $\vu \equiv \vu_{[n]}^{(n+t)}$. We thus split the adjusted teleport vectors into a non-time-dependent component and a time-dependent component: $\mathring\vu_{[n]}^{(n+t)} = \vu + \vw_{[n]}^{(n+t)}$. Summing the time-dependent components for $0 \le t \le \tau - 1$, we introduce:
\[
\vw \equiv \sum_{t=0}^{\tau-1} \vw_{[n]}^{(n+t)} = \alpha \sum_{t=0}^{\tau-1} \sum_{i=1}^{t} \evu_{n+i} \emS^{(n+t)}_{[n], n+i} .
\]
Then, the STI can be expressed as:
\[
\vy = \mM (\tau \vu + \vw),
\]
meaning that $\hat{\vy} \equiv \frac{\vy}{|\vy|}$ can be computed as the PageRank vector of matrix $\mS$ with respect to teleport vector $\frac{\tau}{|\vy|} \vu + \frac{1}{|\vy|} \vw$. Introducing coefficients $ \alpha, \beta, \gamma \in [0,1]$, such that $\alpha + \beta + \gamma = 1$, we can write STI in a general form:
\begin{equation}
\hat\vy = \alpha \mS \hat\vy + \beta \hat\vw + \gamma \hat\vu,
\label{eqn:estimation}
\end{equation}
where $\hat\vw$ and $\hat\vu$ are the normalized vectors of $\vw$ and $\vu$, respectively.

\stitle{Attention.} Because the time-dependent vector $\vw$ is determined by quantities of future states $\mS^{(n+t)}$, we need to estimate it from the known current state $\mS$. A simple way is instead of going $\tau$ time steps in the future, to go $y$ time steps in the past. Assuming teleport probabilities $\evu_{n+t}$ for future papers are equal, we estimate:
\begin{equation}
\tilde\vw \propto \sum_{t=0}^{y-1} \sum_{i=0}^{t} \emS_{[n], n-i} = \sum_{i=0}^{y-1} (y-i) \emS_{[n], n-i}.
\label{eqn:att_estimation}
\end{equation}

We call this estimated vector $\tilde\vw$ the \emph{attention} vector, because for each paper $i$, it computes a weighted count of its citations from the $y$ most recent papers, i.e., its recent attention $\tilde\evw_i \propto y \emS_{i, n} + (y-1) \emS_{i, n-1} + \dots + \emS_{i, n-y+1}$.

Note that if $\vu$ represents the regular teleport probability vector, the attention vector $\vw$ introduces yet another teleport probability, suggesting that researchers tend to to read and cite currently trending papers.
i.e., papers that have recently received significant attention
from the scientific community.
To investigate this hypothesis, we explore
four citation networks (as per the default experimental configuration discussed in Section~\ref{sec:exp:methodology}), and count how many top-$100$ papers were recently popular, based on STI
(i.e., were among the top cited
in the past $5$ years).
As we see in Table~\ref{tab:attention}, roughly half of the top-$100$ papers were, indeed, recently popular.
This observation validates our assumption that the level of attention a paper has recently attracted is indicative of its ability to attract citations in the short-term.

\begin{table}[t]
\centering
\caption{Recently popular papers in top-$100$ (accord. to STI)}
\label{tab:attention}
\begin{tabular}{c|c c c c}
\hline
\textbf{Dataset} & \textbf{hep-th} & \textbf{APS} &
\textbf{PMC} & \textbf{DBLP} \\
\hline
\textbf{Recently Popular} & $41$ & $54$ & $54$ & $63$\\
\hline
\end{tabular}

\end{table}

Attention, however, is not the only mechanism that governs which papers researchers read. Naturally, researchers may read a paper cited in the reference list of another paper. Moreover, similar to previous work \cite{walker2007ranking,sayyadi2009futurerank}, we assume that researchers also read recently published papers. Specifically, we capture the recency of a paper $i$ using a score that decays exponentially based on the paper's age:
\begin{equation}
\label{exponential-factor}
    \evu_i = c e^{\eta \cdot (n - i)},
\end{equation}
where  $n$ is the current time,
$i<n$ denotes the publication time
of paper $ i $, hyperparameter $\eta$ is
a negative constant (as $n-i \geq 0$),
and $c$ is normalization constant so that $|\vu|=1$.
To calculate a proper $\eta$ value, a similar
procedure like the one used in~\cite{sayyadi2009futurerank} can be
followed (see also Section~\ref{param_tuning}).

\stitle{\method.} We refer to the ranking approach based on Equations~\ref{eqn:estimation}, \ref{eqn:att_estimation}, and \ref{exponential-factor} as \method. Note that Eq.~\ref{eqn:estimation} combines three mechanisms.
Specifically, we assume that
a researcher may read 
a paper for one of
the following reasons: the paper gathered
attention recently, the paper was recently published,
or the paper was found in another paper's
reference list. We model this behavior with the following random process.
A researcher, after 
reading paper
$ i$, chooses to read any other paper from $ i $'s
reference list, with probability $ \alpha $. 
With probability $ \beta $ she
chooses a paper
based on its attention.
This behavior essentially makes recently rich papers even richer, and is reminiscent of a
time-restricted preferential attachment mechanism
of the Barab\'asi-Albert model of network
growth~\cite{barabasi2014network}.
Finally, with probability $ \gamma $
she chooses
any paper with a
preference towards recently published ones.

Two special values for coefficient $\beta$ are worth mentioning.
First, observe that when $\beta=0$, a setting we call NO-ATT (for no attention), the model becomes similar to time-aware methods that address the inherent bias against new papers in citation networks (see also Section~\ref{sec-related}, and \cite{KVSDV2019} for a thorough coverage of such approaches). Note that additionally setting $\eta=0$ in Eq.~\ref{exponential-factor} recovers PageRank. Second, when $\beta=1$, a setting we call ATT-ONLY (for attention only), \method is solely based on the attention mechanism, assuming that the recent citation patterns will persist exactly in the near future. To the best of our knowledge, ATT-ONLY has not been considered in the literature as a means to estimate the short-term impact of a paper. As we 
show in Section~\ref{sec-exp}, attention alone is a powerful mechanism, 
often more effective than existing approaches. However, $\beta=1$ is never the optimal setting; it is always better to consider attention in combination with the other two citation mechanisms.

Equation~\ref{eqn:estimation} describes an iterative process for estimating STI vector: starting with a random value, at each step update the vector with the right hand side of the equation. This process is repeated until the values converge. The following theorem, ensures that convergence is achievable.

\begin{theorem}
\label{th:convergence}
The iterative process defined by Eq.~\ref{eqn:estimation} converges.
\end{theorem}

\begin{proof}

We can rewrite
Equation~\ref{eqn:estimation} in matrix form as:
\begin{equation}
\label{paper:vector}
\hat\vy = \mR \hat\vy
\end{equation}
where $\mR$ is a matrix satisfying:

\begin{equation}
    \label{matrix:R}
    \emR_{i,j}  = \alpha \emS_{i,j} 
    + \beta  \hat\evw_i 
    + \gamma \hat\evu_i
\end{equation}

In other words, matrix $ \mR$ is a modified
citation matrix, artificially expanded with
directed edges from any node to any other in the
network.
It is easy to see that it is a stochastic matrix
where each column sums to 1.
Moreover, it satisfies the conditions of irreducibility
and aperiodicity~\cite{langville2011google} and
thus the iterative process is
guaranteed to converge.
\end{proof}


\section{Evaluation}
\label{sec-exp}

This section presents a thorough experimental evaluation of our approach for ranking papers based on their short-term impact. 
Specifically, Section~\ref{sec:exp:methodology} discusses the experimental setup and evaluation approach taken. Section~\ref{param_tuning} investigates the effectiveness of our proposed method and the importance of the attention-based mechanism. Section~\ref{sec:exp:comparison} compares \method with existing approaches from the literature. Finally, Section~\ref{sec:exp:converge} discusses the convergence rate of \method.

\subsection{Experimental Setup}
\label{sec:exp:methodology}

\stitle{Datasets.}
We consider four datasets in our experiments: 

\begin{enumerate}

\item arXiv's high energy physics (hep-th)
collection, which was provided by the 2003 
KDD cup.\footnote{\url{http://www.cs.cornell.edu/projects/kddcup/datasets.html}}
This collection consists of approximately 
27,000 papers with 350,000 references, 
written by 12,000 authors from 1992 to 2003. 
\item A collection of papers provided by 
the American Physical Society 
(APS)\footnote{\url{https://journals.aps.org/about}}, 
which contains about 500,000 papers
with 6 million references, written by 
about 389,000 authors from 1893 to 2014. 
\item A collection of open access papers from 
pubmed central\footnote{\url{https://www.ncbi.nlm.nih.gov/pmc/}} (PMC), which consists of about
1 million papers with 665,000 references,
written by 5 million authors, from 1896 to
2016.
\item A collection of about $3$
million papers and $25$ million references,
written by more than $1.7$ million authors,
from the computer science domain
(DBLP)\footnote{\url{https://aminer.org/citation}}, 
published from 1936 to 2018.
\end{enumerate}

\stitle{Evaluation Methodology.}
To evaluate the effectiveness of 
\method~in ranking papers based on their short-term 
impact, we construct a \emph{current} and a \emph{future} state of the citation network.
We partition each 
dataset according to time in two parts, each having equal number of papers. We use the older half  
to construct the current state of the citation network, denoted as $\mC^{(n)}$. 
All ranking methods will be based on this network acting as the ``training'' subset. 
We use parts of the newer half to construct the future state of the network, denoted as $\mC^{(n+\tau)}$.
All ranking methods will be evaluated based on this network acting as the ``test'' subset.

Specifically, the future state is constructed as follows.
We vary the size, in terms of number of papers, of the future state relative to the size of the current state.
Thus we do not vary the time horizon $\tau$ directly, but rather the \emph{test ratio}, which is
the relative size of the future with respect to the current network.
We consider values for the test ratio among $\{1.2,1.4,1.6,1.8,2.0\}$, where $2.0$ corresponds to using all 
citations in the dataset to define the future state.
In some experiments we fix the test ratio to a default value of $1.6$, meaning that the future state contains $30$\% more papers than the current state.
Table~\ref{ratio-years} presents, for each dataset, 
the length in years of the time horizon that corresponds to each test ratio value.
Note, that the
relationship between test ratio and $\tau$ is not linear, due to  the non-constant number 
of published papers per year and the fact that most datasets contain
incomplete entries for the last year they include.

\begin{table}[!t]
\caption{Correspondence of the Test Ratio to the Time Horizon}
\label{ratio-years}
\centering
\begin{tabular}{c | c c c c}
\hline
\textbf{Test} & \multicolumn{4}{c}{\textbf{Time Horizon $\tau$ (in years)}}\\
\textbf{Ratio} & \textbf{hep-th} & \textbf{APS} & \textbf{PMC} & \textbf{DBLP}\\
\hline
1.2 & 1 & 4 & 1 & 1\\
1.4 & 2 & 7 & 2 & 3\\
1.6 & 3 & 10 & 2 & 4\\
1.8 & 4 & 13 & 3 & 6\\
2.0 & 5 & 16 & 3 & 7\\
\hline
\end{tabular}

\end{table}

Given the future state of the citation network, we can compute the STI of each paper as per its definition (see Section~\ref{sec-back}). Similar to previous approaches in the literature~\cite{sayyadi2009futurerank,walker2007ranking,ghosh2011time,KVSDV2019}, the ranking of papers based on their 
STI forms the \emph{ground truth}. 
Any paper ranking method is oblivious of the future state $\mC^{(n+\tau)}$ of the citation network, and hence the ground truth, and only uses the current state $\mC^{(n)}$ to derive a ranking. To quantify the effectiveness of a method, we compare its produced ranking to the ground truth, using the following two measures:

\begin{itemize}

\item Spearman's $\rho $~\cite{spearman1904proof} is a
non-parametric measure of rank correlation. It is based on 
the $L1$ distance of the ranks of items in two
ranked lists and provides a quantitative measure to 
compare how similar these lists are. Its values range
from $-1$ to $1$ with $1$ denoting perfect
correlation, $-1$ denoting perfect negative
correlation and $0$ denoting no correlation.

\item Normalized Discounted Cumulative Gain at rank $k$ (nDCG@$k$) is a rank-order sensitive metric.
The discounted cumulative gain (DCG) at rank $k$ of a paper is computed as
$\text{DCG}@k = \sum_{i=1}^k{\frac{rel(i)}{\log_2(i+1)}}$, 
where $rel(i)$ is the ground truth score, i.e., 
the short-term impact, of the paper that appears at 
the $i$-th position on the method's ranking. The
nDCG@$k$ is 
the paper's DCG divided by the ideal DCG, achieved
when the method's ranking coincides with the ground
truth. In our evaluation, we consider
values of $k$ 
among $\{5, 10, 50, 100, 500\}$, 
with $k=50$ being used as a default value.

\end{itemize}

Spearman's $\rho $ calculates an overall similarity of the given ranking with the ground truth ranking. In contrast, nDCG@$k$ measures the agreement of the two rankings on the top-ranking papers.

\subsection{Ranking Effectiveness}
\label{param_tuning}

\begin{table}[!t]
\caption{\method's parameterization space.}
\label{tuning_params}
\centering
\begin{tabular}{c | c c c }
\hline
\textbf{Parameter} & \textbf{min} & \textbf{max} & \textbf{step}\\
\hline
$\alpha$ & 0.0 & 0.5 & 0.1\\
$\beta$  & 0.0 & 1.0 & 0.1\\
$\gamma$ & 0.0 & 0.9 & 0.1\\
$y$      & 1   &   5 & 1  \\
\hline
\end{tabular}

\end{table}

In this section, we investigate \method's effectiveness 
for the default experimental 
setting (test ratio equal 
to $1.6$), while varying its parameters, $\alpha, \beta, \gamma$,
and the number $y$ of past years 
used to calculate the attention of a paper.
The range of values tested are shown in Table~\ref{tuning_params}.
For each metric, we discuss \method's parameterization
that achieves the best ranking effectiveness.

First, however, we 
discuss how we set the 
value of the exponential factor
$ \gamma $ of Equation~\ref{exponential-factor}.
We follow a similar approach as 
the one used in~\cite{sayyadi2009futurerank}.
For each dataset, we use an
exponential function of the form
$e^{\widetilde{\gamma}y}$, to fit the 
tail of the distribution of the random 
variable $Y$ that models the probability 
of an article being cited 
$n$ years after its publication.
Figure~\ref{fig:distro}
illustrates the empirical probability 
distribution for each dataset. The factor
$\widetilde{\gamma}$ of the fitting 
function is used as the $\gamma$ value. 
Following this procedure, we calculate
$ \gamma = -0.48$ for hep-th, 
$ \gamma = -0.12 $ for APS and 
$ \gamma = -0.16 $ for PMC and DBLP.

\subsubsection{Effectiveness in terms of Correlation}
\label{correlation-tuning}

\begin{figure}[!t]
\label{heatmaps}
\centering

\begin{subfigure}{0.5\textwidth}
  \includegraphics[width=\textwidth]{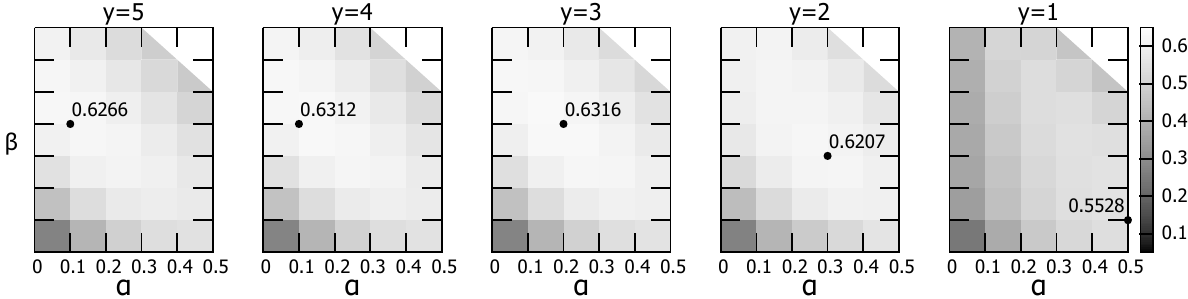}
  \caption{Correlation on the DBLP dataset.
  }
  \label{dblp_correlation_heatmap}
\end{subfigure}

\begin{subfigure}{0.5\textwidth}
  \includegraphics[width=\textwidth]{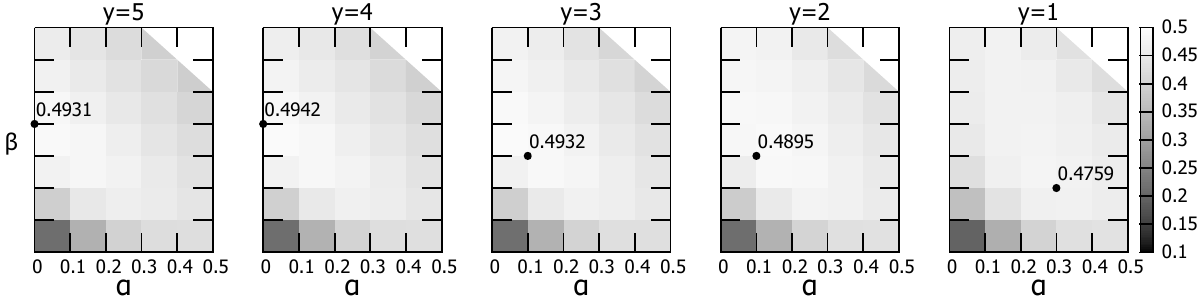}
  \caption{Correlation on the PMC dataset.
  }
  \label{pmc_correlation_heatmap}
\end{subfigure}

\begin{subfigure}{0.5\textwidth}
  \includegraphics[width=\textwidth]{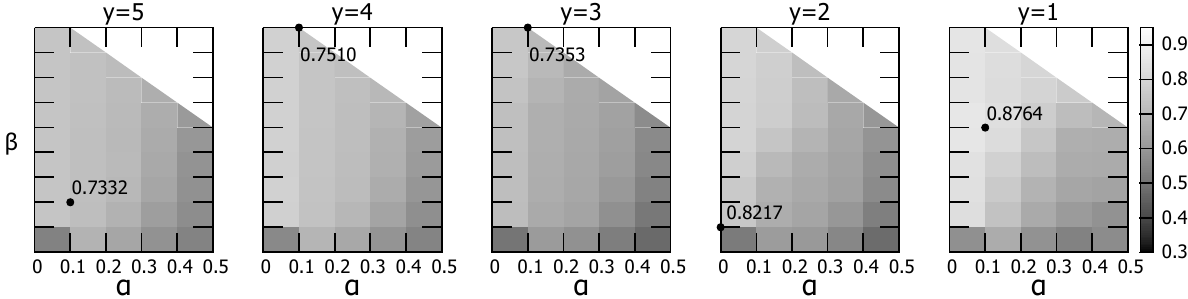}
  \caption{nDCG@50 on the DBLP dataset.
  }
  \label{dblp_ndcg_heatmap}  
\end{subfigure}

\begin{subfigure}{0.5\textwidth}
  \includegraphics[width=\textwidth]{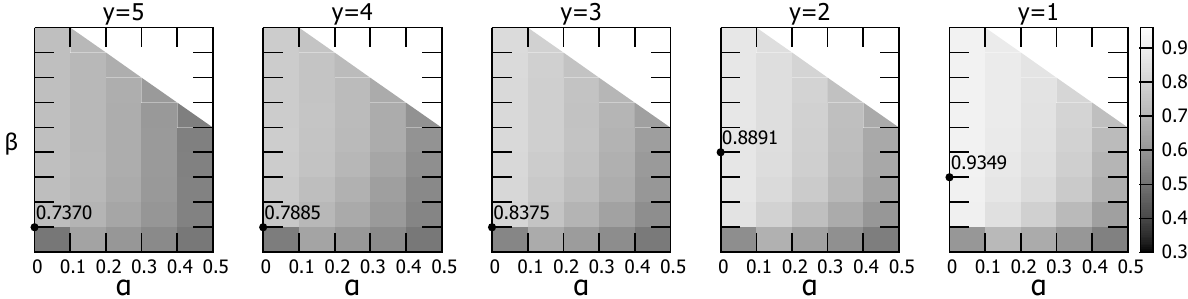}
  \caption{nDCG@50 on the PMC dataset.}
  \label{pmc_ndcg_heatmap}
\end{subfigure}

\caption{Heatmaps depicting the effect of the parameterization of \method to its effectiveness in terms of the correlation and nDCG@50 metrics, for the DBLP and PMC datasets. The best value achieved for each metric is also depicted.
}
\end{figure}

In this experiment, we measure the ranking 
effectiveness of \method, in terms
of Spearman's $\rho$ to the ground truth ranking by STI. 
We can visualize the effectiveness of each tested parameterization as a heatmap over the $\alpha$--$\beta$ space for different values of $y$.
Indicatively, we show the 
heatmaps, for the various parameter settings 
on DBLP and PMC in
in Figures~\ref{dblp_correlation_heatmap}
and \ref{pmc_correlation_heatmap},
respectively.
The heatmaps show
the results varying parameters $\alpha$, and
$\beta$; parameter $\gamma$ is implied since
$\alpha + \beta + \gamma = 1$. 
We expect that as $ \alpha $ increases, \method
simulates researchers that predominantly prefer 
reading papers from reference lists and rarely 
choose papers based on their age, or on whether
they have been recently popular. Thus, as 
$ \alpha $ increases, \method gradually
reduces to simple PageRank, with a small probability 
of random jumps. Since references are made 
only to papers published in the past, researchers
increasingly arrive at older papers when following
references with high probability. As a result, 
for large $\alpha$ values, \method promotes older
papers and, thus, its correlation to the ground truth
is expected to drop. 
Most importantly, the heatmaps validate the 
role of the \ourvec~scores, since for $\beta=0$ (NO-ATT)
we observe significantly lower correlations 
(notice the darker color on the bottom left 
corner of the heatmaps). Similarly, lower correlations
were observed when $\beta=1$ (ATT-ONLY).


From the produced correlation scores, we firstly gather that
\method correlates at least moderately to
the ground truth ranking 
for all datasets in its best setting 
(i.e., $\rho>0.49$).
Further, we observe that the optimal value for
the number of past years $y$, used to calculate the
\ourvec~score 
is $y=3$ and $y=4$ on DBLP and PMC,
respectively, while it's $y=3$ 
on APS, 
and $y=1$ on hep-th.

Interestingly, the first three datasets
follow relatively similar citation patterns (see
Figure~\ref{fig:distro}), with papers
having a citation peak at $2-3$ years after their
publication, while hep-th shows quicker citation 
peaks. Intuitively then, it makes sense to use
a smaller value of $y$ to calculate 
attention scores for hep-th. Since
its research trends may change faster, 
a larger time window to calculate the 
\ourvec would
reflect past research trends, and not current ones. 
On APS, PMC, and DBLP, 
in contrast, papers 
gather citations at a slower rate, thus a larger 
$y$ value is more likely to reflect current 
research preferences.

Based on these experiments, we 
also identify the optimal
parameterization that achieves
maximum correlation
per dataset. We find the best
settings of 
\{$\alpha$,$\beta$,$\gamma$,$y$\} 
at \{$0.3,0.4,0.3,1$\} for hep-th ($\rho=0.6519$),
\{$0.3,0.3,0.4,3$\} for APS ($\rho=0.6295$),
\{$0.0,0.4,0.6,4$\} for PMC ($\rho=0.494$), and
\{$0.2,0.4,0.4,3$\}  for DBLP ($\rho=0.6316$). 
To illustrate the significance
of the attention mechanism, compare these results
to the maximum values 
for $\beta=0$ (NO-ATT). These are 
$0.56$, $0.581$, $0.411$, and $0.529$ 
for hep-th, APS, PMC, and DBLP, respectively.
Accordingly, for $\beta=1$, these values are
$0.615$, $0.537$, $0.444$, $0.571$.

\subsubsection{Effectiveness in terms of nDCG@50}
\label{ndcg-tuning}


We repeat the effectiveness analysis, this time considering the nDCG@50 metric.
Indicatively, we present the heatmaps for DBLP and PMC
in Figures~\ref{dblp_ndcg_heatmap}
and~\ref{pmc_ndcg_heatmap}, respectively. 
An interesting observation is that with 
regards to only capturing the papers with the 
highest short-term impact, smaller time windows
on which the 
\ourvec~scores are calculated seem to be more
suitable. We observe 
that as $y$ increases, the overall nDCG values
decrease  (notice the 
darker hues when  $y$ increases). 
We expect that by 
further increasing $y$, the nDCG would further
drop. This is because by increasing the time
window on which we calculate the 
\ourvec~we, re-introduce the 
inherent age bias of citation networks, 
and the papers with the highest attention scores no longer 
reflect current research trends.
The same observation holds for increased
values of $\alpha$ when $y>1$. 
As $\alpha$ increases,  
the PageRank component dominates \method, giving
advantage to older papers that are not necessarily 
at the current focal point of research. This observation 
is evident from the darker hues on the heatmaps
for values of $\alpha$ close to $0.5$.

Finally, we determine the parameterization that achieves the
best nDCG@$50$ per dataset.
We find the best settings of parameters 
\{$\alpha$,$\beta$,$\gamma$,$y$\} 
at \{$0.0$,$0.4$,$0.6$,$1$\} for hep-th
($nDCG=0.8930$),
\{$0.3$,$0.5$,$0.2$,$3$\} for APS ($nDCG=0.7293$),
\{$0.0$,$0.3$,$0.7$,$1$\} for PMC ($nDCG=0.9349$)
and \{$0.1$, $0.5$, $0.4$, $1$\}  for DBLP ($nDCG=0.8764$).
As before, we observe that the
\ourvec
vector plays a non-negligible role in achieving the 
maximum nDCG on all datasets (i.e., $\beta>0$).
Indicatively, the maximum nDCG@$50$ 
values for
$\beta=0$ are $0.653$, $0.635$, $0.6$, and 
$0.683$ for hep-th, APS, PMC, and DBLP, 
respectively. Accordingly, for $\beta=1$ these
values are $0.89$, $0.71$, $0.930$, $0.862$.

\subsection{Comparative Evaluation}
\label{sec:exp:comparison}

In this section, we compare 
\method to
existing approaches for impact-based paper ranking.
Based on a recent experimental evaluation \cite{KVSDV2019}, we
select the five methods found to be most effective in
ranking by short-term impact.

\begin{itemize}
    \item \textbf{CiteRank (CR)}. This PageRank-based
    method calculates the 
    ``traffic" towards papers by researchers that
    prefer reading recently published 
    papers when performing random jumps~\cite{walker2007ranking}.
    It uses parameters $\alpha\in(0,1)$ and $\tau_{dir}\in(0,\infty)$, where $\alpha$ models 
    the probability with which
    researchers follow references from papers 
    they read and $\tau_{dir}$ models an 
    aging factor, which determines 
    the papers which
    random researchers are more likely to select
    when performing random jumps. In the original 
    work, their optimal settings are found for $\{\alpha,\tau_{dir}\}$ set to $\{0.48,1\}$,
    $\{0.5,2.6\}$,$\{0.31,1.6\}$,$\{0.55,8\}$.
    
    \item \textbf{FutureRank (FR)}. This method is
    based on PageRank and
    HITS~\cite{kleinberg1999authoritative}. 
    It applies mutual reinforcement from papers to
    authors and vice versa, while additionally using
    time-based weights to promote recently published
    papers~\cite{sayyadi2009futurerank}. 
    It uses four parameters: 
    $\alpha,\beta,\gamma\in(0,1)$, and
    $\rho\in(-\infty,0)$. 
    Parameter $\alpha$ is taken from PageRank, 
    $\beta$ is the coefficient of an author-based
    score vector, and $\gamma$ is the coefficient 
    of time-based weights. These
    weights depend on $\rho$, which modifies an
    exponentially decreasing function. In the 
    original work the optimal settings
    of $\{\alpha,\beta,\gamma,\rho\} $ are 
    $\{0.4,0.1,0.5,-0.62\}$, and $\{0.19,0.02,0.79,-0.62\}$.

    \item \textbf{Retained Adjacency Matrix (RAM).} This citation count variant uses a citation age-weighted adjacency matrix \cite{ghosh2011time}. It uses a parameter 
    $\gamma\in(0,1)$ as the base of an exponential
    function, to modify citation weights,
    based on their age. The authors find 
    $\gamma \in \{0.3,0.6,0.71\}$ as as the optimal
    settings.
    
    \item \textbf{Effective Contagion Matrix (ECM).} This method, based on Katz centrality,
    operates over 
    a citation age-weighted adjacency matrix \cite{ghosh2011time} and calculates weights
    of citation chains. It uses parameters,
    $\alpha,\gamma\in(0,1)$, where $\gamma$ is
    taken from RAM, and $\alpha$ is used
    to decrese citation chain weights
    as they increase in length. In the
    original work, the authors find the best settings
    of  $\{\alpha,\gamma\}$ to be $\{0.1,0.3\}$ or $\{0.007,0.71\}$.
    
    \item \textbf{WSDM cup's 2016 
    winner (WSDM)}. We consider the winning 
    solution \cite{feng2016efficient} of a scholarly 
    article ranking challenge.
    This method uses three bipartite networks (papers-authors,
    papers-papers, and papers-venues). It calculates paper
    scores by aggregating scores propagated to papers by other 
    papers, by their authors, and their venues, 
    additionally 
    using 
    scores based on paper in- and out-degrees.
    Paper scores are calculated iteratively,
    based on a fixed small number of
    iterations.
    The method uses 
    parameters $\alpha,\beta \in \mathbb{R}$, 
    as coefficients of each paper's
    in- and out-degree, to calculate 
    paper scores, and the number
    of iterations, $i$. 
    The authors use $\{\alpha,\beta\}=\{1.7,3\}$ in their work
    and set $i \in \{4,5\}$.
    
\end{itemize}

\begin{table}[!t]
\centering
\footnotesize
\caption{Parameterization space of competitors.}
\label{rival_params}
\begin{tabular}{c | c c c c }
\hline
\textbf{Method} & \textbf{Parameter} & \textbf{min} & \textbf{max}
& \textbf{step}\\
\hline
\multirow{ 2}{*}{\textbf{CR}} & $\alpha$ & 0.1 & 0.7 & 0.2\\
  & $\tau_{dir}$ & 2 & 10 & 2\\
\hline
\multirow{ 4}{*}{\textbf{FR}} & $\alpha$ & 0.1 & 0.5 & 0.1\\
& $\beta$ & 0.0 & 0.9 & 0.1\\
& $\gamma$ & 0.0 & 0.9 & 0.1\\
& $\rho$ & -0.82 & -0.42 & 0.2\\
\hline
\textbf{RAM} & $\gamma$ & 0.1 & 0.9 & 0.1\\
\hline
\multirow{ 2}{*}{\textbf{ECM}} & $\alpha$ & 0.1 & 0.5 & 0.1\\
\cline{2-5}
& $\gamma$ & 0.1 & 0.5 & 0.1\\
\hline
\multirow{ 3}{*}{\textbf{WSDM}} & $\alpha$ & 1.1 & 2.3 & 0.3\\
& $\beta$ & 1 & 5 & 1\\
& $i$ & 4 & 5 & 1\\
\hline
\end{tabular}
\end{table}



The optimal parameterization for the competitors
is presented in each 
work. However, these suggested values result 
from the use of 
particular datasets and specific 
experimental settings, which
differ among works. Therefore, 
in our evaluation, we extensively tuned all
competitors, to ensure a fair comparison of 
their effectiveness in ranking based on STI.
Table~\ref{rival_params} 
presents the 
examined parameter sets.\footnote{The 
settings were chosen so as to include, 
for each parameter, values close, or equal, to those suggested 
in the original works. Note, that since 
the value of one parameter may restrict the range of 
others, the total number of settings does not equal 
the sum of all individual parameter settings. Note 
also, that some works do not provide a formal proof 
of convergence. Hence, we exclude the parameter 
ranges in Table~\ref{rival_params} which resulted in non-convergence.}
In total, we used $20$ different settings for
CR, $120$ settings for FR, $9$ settings for RAM,
$25$ settings for ECM, and $50$ settings for WSDM.
Note, that since WSDM requires venue data, we ran this 
method only on the PMC and DBLP datasets, for which  this data was available.
Further, we ran all iterative methods until
the convergence 
error drops below $ 10^{-12} $, 
to ensure that all 
scores approach their final values and
further iterations are not expected to
change the ranking of papers.

In addition to these existing approaches, we also consider two variants of \method that better demonstrate the effect of the attention mechanism. The first, denoted as NO-ATT, 
removes the attention mechanism in \method, i.e., sets $\beta = 0$. Conversely, the second, denoted ATT-ONLY, considers only the attention mechanism in \method, i.e., sets $\beta = 1$.

\begin{figure}
  \centering
\includegraphics[width=0.9\linewidth]{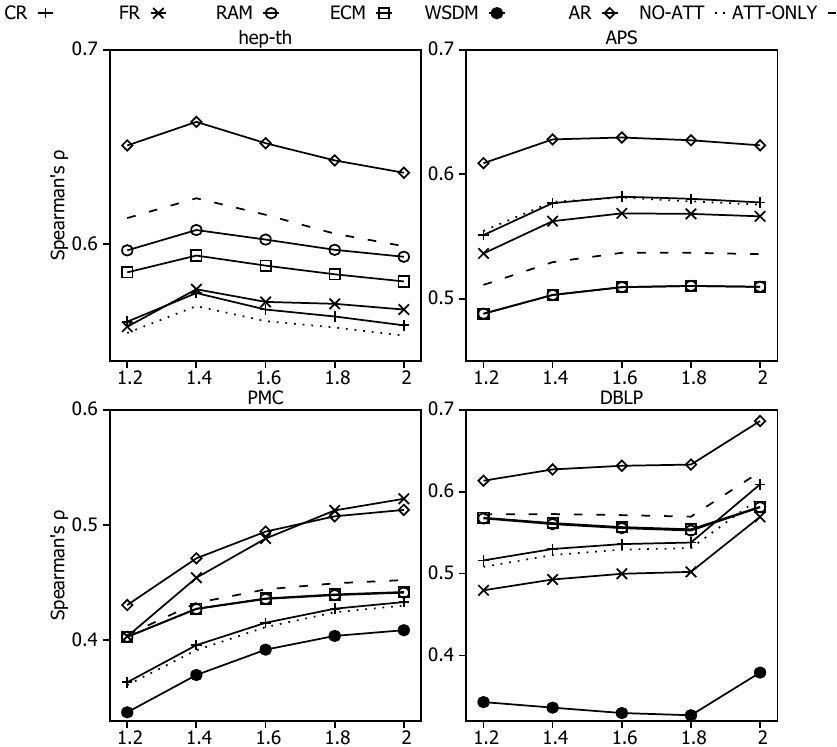}
  \caption{Effectiveness of all methods in terms of correlation.
  The x-axis represents the varying test ratio. }
  \label{comparative_correlation}
\end{figure}

\subsubsection{Effectiveness in terms of Correlation}
\label{exp:correlation}

In this experiment, we measure the correlation of 
each method's ranking to that of the ground truth.
We vary the test ratio of the size of networks 
according to 
Section~\ref{sec:exp:methodology}. For each dataset and test ratio,
we choose the parameterization with the 
best correlation. 
Figure~\ref{comparative_correlation} presents the results.

We observe that \method's ranking better correlates
to the ground truth ranking, compared to all 
competitors on all settings for the hep-th, APS, 
and DBLP datasets. In particular, \method
increases correlation by
up to $0.055$ units on hep-th, by up to $0.057$ on APS, and by
up to $0.079$ on DBLP 
with respect to the best competitor. 
Further, on most
settings \method correlates better 
to the ground truth ranking on PMC, by up to $ 0.027 $, 
compared to the best competitor, while marginally 
losing to FR on two settings (the correlation values 
observed differ by $< 0.01$). It is worth highlighting
that FR achieves such a good correlation only for PMC;
in the other datasets, it 
is outperformed by other existing methods. 
In contrast, \method is robust across
datasets and settings with a large correlation gain 
over all competitors (except FR in PMC).

Our method's performance can be attributed to the 
fact that, compared to the time-aware competitors,
it does not simply promote papers recently cited, 
or recently published. Instead, because of the attention
mechanism, it heavily promotes
 well-cited, recent papers, compared 
to lesser cited recent papers. 
As discussed in 
Section~\ref{sec-method}, recently popular papers
indeed remain popular.
Moreover, our method promotes older papers that are still heavily cited. 
The importance of the attention mechanism is illustrated
by the fact that in two datasets ATT-ONLY outperforms
existing methods. Turning off attention completely, i.e, 
the NO-ATT method,
results in subpar performance, except in one dataset.
Most importantly, in all cases,
the effectiveness is increased when 
the attention mechanism is balanced with the other mechanisms
in \method.


\begin{figure}[!t]
  \centering
 \includegraphics[width=0.9\linewidth]{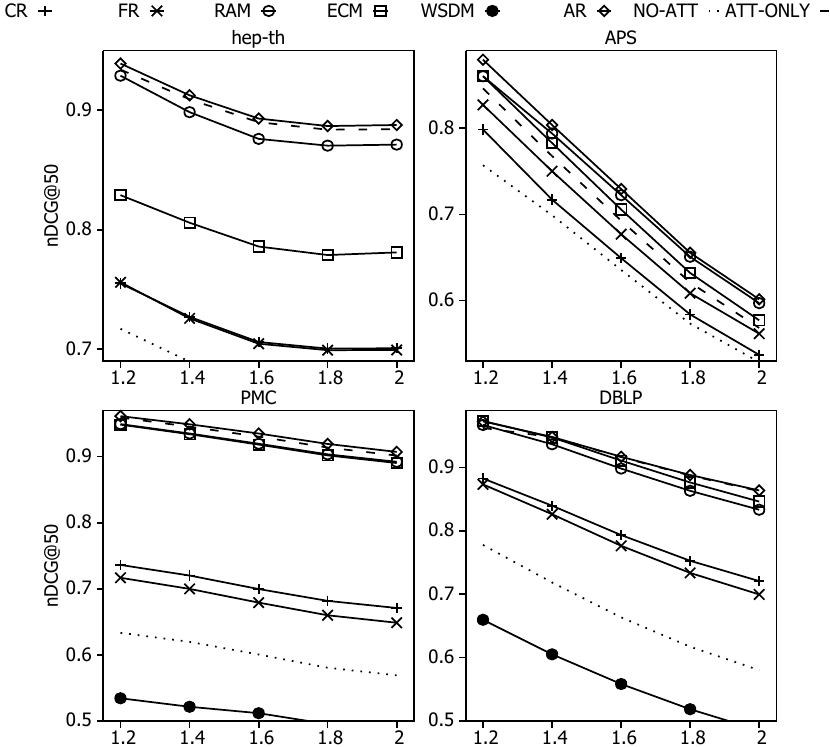}
  \caption{Effectiveness of all methods in terms of nDCG@50. 
  The x-axis represents
  the varying test ratio. }
  \label{comparative_ndcg_varying_fcc}
\end{figure}

\subsubsection{Effectiveness in terms of
nDCG}


In this section, we measure the nDCG 
achieved by each method with regards 
to the ground truth.
We conduct two experiments: in the first,
we set $k=50$ as the cut-off rank when
computing nDCG, varying the test ratio.
In the second experiment we use the
default test ratio (at $1.6$) and measure
nDCG varying $k$.

Figure~\ref{comparative_ndcg_varying_fcc}
presents the results varying the test ratio.
For each setting, 
we select the parameterization of
each method that gives the best nDCG@$50$
value.
In general, as we look further into the
future, i.e., increase the test ratio, 
the ranking accuracy of all methods drops;
the effect is more pronounced in the APS
dataset, and less in hep-th. In all cases,
\method outperforms all competitors,
except for one setting on
DBLP (for $n=1.2$, it marginally 
loses to ECM, with a difference of $<0.001$). 
In  particular, our method improves nDCG@$50$ 
by up to $ 0.017 $ units on hepth, 
$ 0.018 $ on APS, 
$0.0156$ on PMC, and  $ 0.0173 $ on DBLP,
compared to the best existing method.
It is worth mentioning that the best 
existing method varies across datasets, 
being either RAM or ECM.

Figure~\ref{comparative_ndcg_varying_k_best}
presents the results varying $k$ for
a test ratio of $1.6$. For each setting, 
we select the parameterization of each
method that gives the best nDCG@$k$ value.
In general we observe that 
\method is at least on par, and
mostly outperforms all rivals on 
all datasets, with the sole exception 
of nDCG@$5$ on APS (the measured 
difference compared to 
the best competitor is $0.015$). 
Specifically, \method achieves a higher
nDCG value of up to $0.017$ units on hep-th,
up to $0.013$ units on APS
(except nDCG@$5$), up to $0.0019$ on
PMC, and up to $0.01$ units on DBLP.
Additionally, for small values of $k$ 
($k=\{5,10\}$) \method achieves nDCG values
close to $1$ on two out of four datasets
(hep-th and PMC). 
The best competitors 
are again RAM and ECM, depending on the
dataset.

Regarding the special cases of \method, we observe in both 
Figures~\ref{comparative_ndcg_varying_fcc}
and~\ref{comparative_ndcg_varying_k_best},
that excluding attention (NO-ATT) results 
in a significant drop in nDCG. On the
other hand, attention alone (ATT-ONLY)
outperforms most existing methods except 
in APS. As also observed in the case of
Figure~\ref{comparative_correlation},
carefully balancing the mechanisms in 
\method leads to a considerable 
improvement in ranking accuracy.


\begin{figure}[!t]
  \centering
  \includegraphics[width=0.9\linewidth]{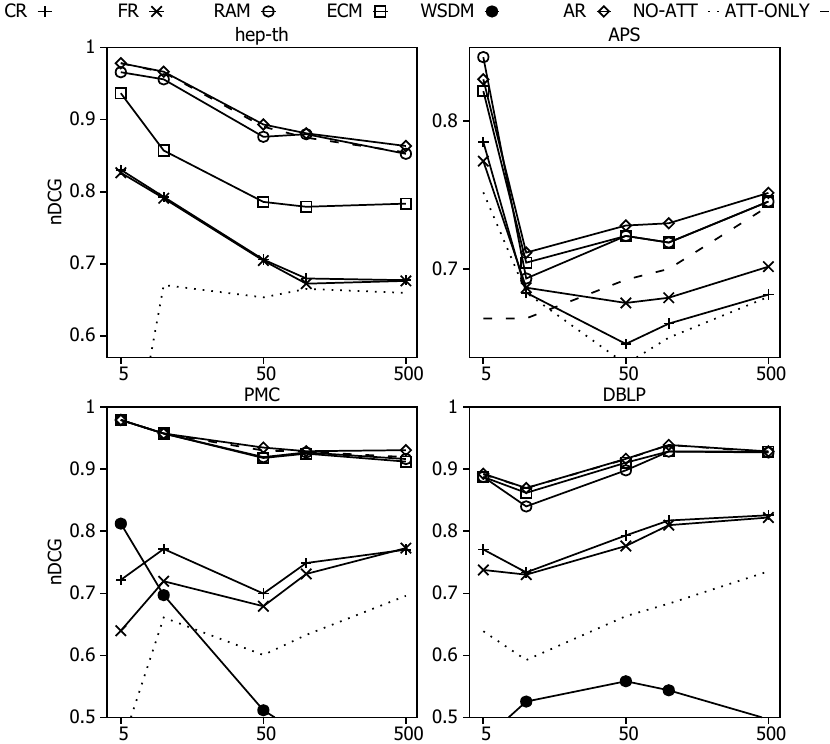}
  \caption{Effectiveness of all methods in terms of nDCG@$k$ on the default test
  ratio. The x-axis represents
  the varying value of $k$.
  }
  \label{comparative_ndcg_varying_k_best}
\end{figure}

\subsection{Convergence of AttRank}
\label{sec:exp:converge}

\method involves an iterative process, similar to PageRank, to compute scores for papers. 
Specifically, we can view \method as a PageRank variant,
where the random jump vector is replaced by two
vectors, the attention-based vector and
the publication age-based vector, and thus PageRank's random
jump probability $ 1-\alpha $ is divided among $\beta,\gamma$ in \method.
The convergence of \method is thus affected by the same factors as PageRank's; an in-depth discussion of
PageRank's convergence properties can be found 
in~\cite{langville2011google}.
The most important property is that as $ \alpha \to 1$, the
convergence rate decreases and more iterations are required.

Following the discussion in Section~\ref{sec-exp}, however,
large values of parameter $ \alpha $ do not favor ranking
based on short-term impact, and 
\method's optimal effectiveness is 
always achieved for $ \alpha \leq 0.5 $.
Additionally, as $ \alpha \to 0$, \method
tends to 
depend increasingly on 
the sum of the attention- and time-based 
vectors. Thus, the number of iterations 
required for convergence decreases, with the limit case $ \alpha = 0$ requiring 
a single iteration (i.e.,
the calculation of the attention- and time-based vectors).

Overall, \method is expected to converge faster than 
Page\-Rank and other variants (PageRank has been used with $ \alpha = 0.5$ on 
citation networks~\cite{chen2007finding,ma2008bringing}).
In our experiments, \method converges in less than $30$ iterations
for hep-th, APS, and DBLP, and less than $20$ iterations for PMC,
for $ \alpha = 0.5 $ and a convergence error 
of $ \epsilon \leq 10^{-12}$, with the number of iterations
decreasing for smaller values of $\alpha$. Compare this to
the maximum required iterations for CR, 
which are $51$, $46$, $26$, and $47$, for hep-th, APS, PMC,
and DBLP, respectively, for $ \alpha = 0.5 $. 
The corresponding numbers for FR (which 
did not, in practice, converge under all
possible settings) are 
$ 35 $, $ 30 $, $26$, and $ 23 $, for hep-th, APS, PMC,
and DBLP, respectively, and for $ \alpha = 0.5 $.

\section{Related work}
\label{sec-related}

In recent years, various methods have been proposed for quantifying the scientific impact of papers. In the following, we review the most important work, focusing on methods to rank papers by their expected short-term impact.
For a thorough coverage of this research area refer to \cite{BaiLZNKLX17,KVSDV2019}.

\stitle{Basic Centrality Variants.}
A large number of methods are Page\-Rank adaptations tailored to better 
simulate the way a random researcher traverses the citation 
network while reading papers (e.g.,~\cite{yao2014ranking,zhou2016ranking,sidiropoulos2006generalized,chen2007finding}).
While such approaches modify the random
researcher's behaviour in intuitive ways (e.g., she 
prefers reading cited papers that are similar to the one 
she currently reads), they do not address 
age bias, an important intrinsic issue in citation networks.

\stitle{Time-Aware Methods.}
To alleviate age bias, a number of time-aware methods were
proposed. These methods introduce time-based weights in 
the various centrality metric calculations, to 
favor either recent 
publications (e.g.,~\cite{yu2005adding,dunaiski2012comparing,walker2007ranking,sayyadi2009futurerank}
or recent citations (e.g.,~\cite{ghosh2011time}), or citations received 
shortly after the publication of an article (e.g.,~\cite{yan2010weighted,zhang2018ranking}).\footnote{
Note that time-aware weights with different interpretations have been proposed in a limited number of works, in particular in~\cite{hsu2016timeaware,ma2018query}.
}

Although the aforementioned practice has been applied to citation count 
variants~\cite{yan2010weighted,zhang2018ranking,ghosh2011time} or to 
Katz centrality~\cite{ghosh2011time}, most works introduce
time-awareness to PageRank adaptations. This 
is achieved by
modifying either the adjacency
matrix~\cite{yu2005adding,dunaiski2012comparing}
and/or landing probabilities in the Page\-Rank formula
(e.g.,~\cite{walker2007ranking,sayyadi2009futurerank,hwang2010yet,dunaiski2012comparing}). 
In the former case, the intuition is that the 
random researcher avoids following references to old papers (with respect to the current time or to the publication year of the citing paper). In the latter case, the random researcher 
prefers selecting new papers during random jumps.

Time-awareness is shown to improve the accuracy when
ranking by short-term impact. However, it fails to
differentiate among recent papers favoring all
equally. In reality, some papers are fitter than
others and will attract more attention. To address
this issue, literature proposes using additional
information besides the citation network, such as
paper metadata and other networks.

\stitle{Metadata.}
An interesting approach is to incorporate paper metadata (e.g., information about authors, venues) into the ranking method. 
Scores based on these metadata can be derived either
through simple statistics calculated on paper scores
(e.g., average paper scores for authors or venues),
or from well-established measures such as the
Journal Impact Factor~\cite{garfield2006history},
or the Eigenfactor~\cite{bergstrom2008eigenfactor}.
The majority of approaches in this category
incorporates paper metadata
in Page\-Rank-like models, to modify citation/transition
matrices (e.g.,~\cite{yan2010weighted}), or both citation/transition matrices and random jump probabilities~\cite{hwang2010yet,hsu2016timeaware}. 
An alternative to the above approaches is presented 
in~\cite{yu2005adding}
which calculates the scores of recent papers, for
which only limited citation information is currently
available, solely based on metadata, while using a
time-aware PageRank model for the rest. 

\stitle{Multiple Networks.}
Another way to incorporate additional information is to
define iterative processes on 
multiple interconnected networks (e.g., author-paper, 
venue-paper networks) in addition to the basic citation 
network. 
We can broadly discern two approaches: the first 
is based on mutual reinforcement, an idea 
originating from
HITS~\cite{kleinberg1999authoritative}. 
Methods following this approach (e.g.,~\cite{sayyadi2009futurerank,wang2013ranking})
perform calculations on
bipartite graphs where nodes on either side of the 
graph mutually reinforce each other (e.g., paper scores
are used to calculate author scores and vice versa), in
addition to calculations on homogeneous networks (e.g. 
paper-paper, author-author). In the second approach, 
a single graph spanning heterogeneous nodes is used for all 
calculations~\cite{nie2005object,jiang2012towards} 
and scores are propagated between all types of nodes during 
an iterative process.

\stitle{Ensemble Techniques.}
A popular approach for improving ranking accuracy is to consider ensembles that combine the rankings from multiple methods.
The majority of the 2016 WSDM 
Cup\footnote{The task was to rank papers based on their ``query-independent importance'' using information from multiple interconnected networks~\cite{wade2016wsdm}.}
paper ranking methods (e.g.~\cite{feng2016efficient,chang2016ensembleofranking,wesley2016static})
and their extensions (like~\cite{ma2018query}) fall in this category. They combine several types of
scores like in- and out-degrees, simple and time-aware PageRank scores, metadata-based scores etc., calculated on different graphs (citation network, co-authorship network, etc). 
For instance, 
the winning solution of the cup~\cite{feng2016efficient} 
(see Section~\ref{sec-exp}), 
combines various scores derived from in- and out-degrees
with scores propagated from venues, papers, and authors.

\stitle{Paper Citation Prediction.}
A separate line of work is concerned with modeling the arrivals of citations for \emph{individual} papers 
to predict their \emph{long-term} impact. Early approaches \cite{yan2011citation,DavletovAC14} model the problem as a time series prediction task. Following the seminal work of \cite{wang2013quantifying}, 
subsequent works model the arrival of citations using non-homogeneous Poisson \cite{shen2014modeling} or Hawkes \cite{XiaoYLJWYCZ16} processes.

This line of work is ill-suited for ranking by short-term impact for two reasons. First, it has a different goal, predicting the citation trajectory of individual papers, and as such it optimizes for the prediction error with respect to the actual citation trajectories. Second, the training process is prone to overfitting \cite{wang2014comment}, and requires a long history ($\geq 5$ years) of observed citations for each paper. In constrast, the 
majority of the top ranking papers by short-term impact are recent publications. 
For example, in the default experimental configuration of the PMC dataset (see Section~\ref{sec:exp:methodology}) $79\%$ of the top-$100$ papers are published in the last $5$ years.

\stitle{Discussion.}
The time-awareness mechanism is not sufficient for distinguishing the short-term impact of papers. As explained, recent work focuses on using additional data sources  (venues, co-authorship networks, etc.) to build better informed models. However, an important limitation of this strategy is that this data is not readily available, fragmented in different datasets, not easy to collect, integrate and clean, and is often incomplete. In contrast, our approach is to rely solely on the properties of the underlying citation network, and try to better model the process with which the network evolves.

\section{Conclusion}
\label{sec-concl}

In this work, we present 
\method, a method that effectively ranks papers 
based on their expected short-term impact.
The key idea is to carefully utilize the recent attention a paper 
has received.
Specifically, our method models the process of a random researcher reading papers from the literature, and incorporates an 
\ourvec mechanism to identify popular papers that are likely to 
continue receiving citations, as well as a time-based mechanism 
to promote recently published papers that have not yet 
received sufficient citations.

We studied the effectiveness of our approach in terms of Spearman's rank
correlation and nDCG compared to the ground 
truth rankings compiled from the short-term impact of papers across four different citation networks.
Our findings demonstrate that our
method outperforms existing methods in terms of both
metrics.
Moreover, they validate the introduction of the attention-based
mechanism. The effectiveness of our approach degrades when the attention-based mechanism is completely removed, or when used in isolation.


\bibliographystyle{IEEEtran}

\bibliography{bibfile.bib}

\end{document}